\documentclass[11pt]{article}
\usepackage{amsfonts,amsmath,amssymb,amsthm}
\usepackage[pdftex]{graphicx}
\usepackage[hmargin=1in,vmargin=1in]{geometry}
\usepackage{natbib}
\usepackage[colorlinks,citecolor=blue]{hyperref}
\usepackage{enumerate}
\usepackage{caption}
\usepackage[finnish, british]{babel}
\usepackage{capt-of}
\usepackage{hyperref}
\usepackage[flushleft]{threeparttable}
\def\qed{\rule{2mm}{2mm}}

\let\footnote=\endnote

\usepackage{tikz,graphicx,pgfplots,pgfkeys,subcaption}   %The subcaption package is used in place of "subfig"
\def\addlegendimage{\csname pgfplots@addlegendimage\endcsname}

\usepackage[pagewise,mathlines]{lineno}
\mathchardef\dash="2D
\newtheorem{theorem}{Theorem}[section]
\newtheorem{lemma}{Lemma}[section]

\newtheorem{assumption}{Assumption}[section]
\theoremstyle{definition}

\newtheorem{remark}{Remark}[section]
\usepackage{etoolbox} % This package goes here to add \qed to remarks automatically.
\AtEndEnvironment{remark}{~\qed}%
\AtEndEnvironment{example}{~\qed}%

\begin{document}
%\date{}

\author{
Vishal Kamat\footnote{I am grateful to Ivan Canay for his valuable guidance and suggestions. I thank the Co-Editor, two anonymous referees, Matias Cattaneo, Joel Horowitz, Pedro Sant'Anna and Max Tabord-Meehan  for their helpful comments.}\\
Departments of Economics\\
Northwestern University\\
\url{v.kamat@u.northwestern.edu}
}

\title{On Nonparametric Inference in the Regression Discontinuity Design}

\maketitle

\begin{abstract}
This paper studies the validity of nonparametric tests used in the regression discontinuity design. The null hypothesis of interest is that the average treatment effect at the threshold in the so-called sharp design equals a pre-specified value. We first show that, under assumptions used in the majority of the literature, for \emph{any} test the power against any alternative is bounded above by its size. This result implies that, under these assumptions, any test with nontrivial power will exhibit size distortions. We next provide a sufficient strengthening of the standard assumptions under which we show that a novel test in the literature can control limiting size.
\end{abstract}



\section{Introduction}

The nonparametric literature on the regression discontinuity design (RDD) is characterized by the nonparametric identification of parameters at the threshold. In this paper we study constructing tests for these parameters, for which numerous alternatives are present in econometrics - see, for example, \cite{Mc08}, \cite{FFM12}, \cite{CCT14} and \cite{OXM15} for such tests, and see \cite{LL10} and \cite{IL08} for recent surveys on the literature. In particular, we focus on the null hypothesis that the average treatment effect at the threshold in the sharp design equals a pre-specified value.

When testing this null hypothesis in simulation studies (not reported), we observe that available tests fail to control the rejection probability under some null distributions with practical sample sizes. This failure occurs for distributions that satisfy the typically imposed assumptions, and in turn makes us question the reliability of current inference procedures. Here we hence formally study the construction of valid tests for our null hypothesis. As stated in Section \ref{sec:uniform}, the aim is to ideally construct a finite sample valid test, which requires the finite sample control of size, i.e. the null rejection probabilities. Since nontrivially achieving this may be too demanding, one may aim to approximate this finite sample goal in large samples through two different definitions of asymptotic validity. The first termed uniform asymptotic validity requires limiting control of null rejection probability uniformly across distributions under the null, whereas the second termed pointwise asymptotic validity requires such control to hold for each fixed distribution under the null. As highlighted in Remark \ref{rem:point}, current tests are shown to only satisfy the second definition, which may not provide any guarantee on the control of finite sample size. The practical importance of the distinction in these definitions has also been previously noted in various other econometric applications - see, for example, \cite{M07} and \cite{M12} for unit roots in autoregressive models, \cite{RS08} and \cite{DG09} for moment inequality models, \cite{LP05} and \cite{AG09} for post model selection and \cite{D97} and \cite{M10} for weak instrumental variable models. 

Our first result establishes that, under standard assumptions in the basic setup, any test for our null hypothesis of interest will have power against any alternative bounded above by its size. This implies that it is impossible to construct nontrivial finite sample valid tests and uniformly asymptotically valid tests under these assumptions. Intuitively, this result occurs because the assumptions permit a set of possible distributions that is `too large', in a sense made precise in Lemma \ref{lem:ML}. This causes distributions under the null and alternative to be `arbitrarily' close making it impossible to distinguish them given the data. Our goal through this impossibility result is not to criticize current nonparametric tests but to attempt to caution researchers using them. Such nonparametric tests are often viewed as appealing as they only require imposing mild regularity assumptions. We hope to convey that these assumptions however allow the permitted set of distributions to be arbitrarily large resulting in misleading inference. To recover reliable inference, the researcher would then naturally need to strengthen the assumptions to further restrict the permitted set of distributions. To this end, our second result illustrates a sufficient strengthening of the standard assumptions under which the \cite{CCT14} test is uniformly asymptotically valid. Our stronger assumptions are analogous to the ones commonly required for optimality results in nonparametric estimation; see, for example, \citet[][Chapter 24]{V98}.

In addition to the literature on RDD, this paper is also related to the growing one in econometrics on the testability of hypotheses. \cite{BS56} was the initial paper to demonstrate the impossibility of constructing nontrivial valid tests for the mean of a distribution. \cite{R04} extended this result to provide sufficient conditions to examine the testability of hypotheses in different settings. The key insight is formalizing the notion of closeness of the set of null and alternative distributions using the total variation metric. In econometrics, \cite{CSS13} verified one of these conditions to establish impossibility of constructing nontrivial valid tests for some hypotheses in nonparametric models with endogeneity. In this paper we verify the same condition, restated as Lemma \ref{lem:ML} here, to prove our impossibility result. Alternatively, \cite{G10,G10b} used a direct approach of considering sequences of distributions under the null to show limiting size distortions in the Hausman pretest. For further examples of such impossibility results see \cite{LL90}, \cite{LP08} and \cite{U08}, and for a review of such results in econometrics see \cite{D03}.

The remainder of the paper is organized as follows. Section \ref{sec:setup} describes the basic RDD setup, where we introduce the notation, the commonly imposed assumptions and the null hypothesis of interest. Section \ref{sec:uniform} states our testing problem. Section \ref{sec:results} illustrates our main results.

\section{Basic RDD Setup}\label{sec:setup}

Assume there are random variables $(Y(0),Y(1),Z)\sim Q \in \mathbf Q$, where $\mathbf{Q}$ is a set of distributions on a sample space $\mathcal{W} = \mathcal{Y} \times \mathcal{Y} \times \mathcal{Z} \subseteq \mathbf{R} \times \mathbf{R} \times \mathbf{R}$ such that $\mathcal{Z}$ contains a neighbourhood of zero. Here, let $Y(0)$ denote the potential outcome under treatment zero, $Y(1)$ denote the potential outcome under treatment one, and $Z$ denote an observed predetermined characteristic. The observed random variables from the experiment are $(Y,Z) \sim P \in \mathbf{P}$, where $\mathbf{P}$ is a set of distributions on a sample space $\mathcal{X} = \mathcal{Y} \times \mathcal{Z} \subseteq \mathbf{R} \times \mathbf{R}$. The observed outcome is determined by
\begin{align}\label{eq:RDDE}
Y = A \cdot Y(1) + (1 - A) \cdot Y(0)~,
\end{align}
where treatment assignment follows a normalized threshold rule of the form
\begin{align}\label{eq:RDDA}
A = 1 \{ Z \geq 0 \}~.
\end{align}
Since
\begin{align}
  (Y, Z) = M(Y(0),Y(1),Z)~,
\end{align}
where  $M:\mathcal{W} \to \mathcal{X}$ is the mapping implied by \eqref{eq:RDDE}, we have that $P = QM^{-1}$ and
\begin{align}\label{eq:push}
  \mathbf{P} = \{QM^{-1} : Q \in \mathbf{Q}\}~,
\end{align}
where  $M^{-1}$ is the pre-image of $M$. Let $W^{(n)}=\{(Y_i(0),Y_i(1),Z_i):1\le i\le n\}$ denote an i.i.d sample from $Q$, and let $X^{(n)}=\{(Y_i, Z_i) : 1 \leq i \leq n \}$ denote the corresponding observed i.i.d sample from $P$. Further, let $P^n$ denote the $n$-fold product $\bigotimes_{i=1}^n P$, i.e. the joint distribution of the observed data.

We next illustrate the standard assumptions and the resulting set of possible distributions $\mathbf{Q}$, which plays a fundamental role in our analysis. In order to do so, we introduce further notation. Let $\mu_{-}(z,Q)=E_{Q}[Y(0) | Z=z]$ and $\mu_{+}(z,Q) = E_{Q}[Y(1) | Z=z]$, and, whenever $\mu_{-}(\cdot,Q)$ and $\mu_{+}(\cdot,Q)$ have the appropriate level of differentiability,  let $\mu_{-}^{v}(z,Q)= d^{v} \mu_{-}(z,Q) / dz^{v}$ and $\mu_{+}^{v}(z,Q)= d^{v} \mu_{+}(z,Q) / dz^{v}$. Further, let $\sigma_{-}^{2}(z,Q) = Var_Q[Y(0) | Z=z]$ and $\sigma_{+}^{2}(z,Q) = Var_Q[Y(1) | Z=z]$, and let $f_{Q}(z)$ denote the density of $Z$. Using this notation, let 
\begin{equation}\label{eq:Qdef}
  \mathbf{Q} = \{ Q \in \mathbf{Q}_{\mathcal{W}} : Q \text{ satisfies Assumption \ref{ass:PC}} \}~,
\end{equation}
where $\mathbf{Q}_{\mathcal{W}}$ denotes the set of all Borel probability measures on $\mathcal{W}$ that have a density on $Z$ with respect to the Lebesgue measure, and Assumption \ref{ass:PC} is stated below. Assumption \ref{ass:PC}, in particular, captures the commonly imposed restrictions in the majority of the nonparametric RDD literature; see, for example, \cite{CCT14} and \cite{IK13}.

\begin{assumption}\label{ass:PC}
Let $Q$ be such that there exist real numbers $\kappa(Q) > 0$,  $L(Q) > 0$ and $U(Q) > 0$ where for all $z \in (-\kappa(Q),\kappa(Q))$, i.e. in a neighbourhood around the threshold, the following conditions hold true.
\begin{enumerate}[(i)]
\item $f_Q(z)$ is continuous and $ L(Q) \leq f_Q(z) \leq U(Q)$.
\item $E_{Q}\left[Y(0)^{4} | Z=z \right] \leq U(Q)$ and $E_{Q}\left[Y(1)^{4} | Z=z \right] \leq U(Q)$.
\item $\mu_{-}(z,Q)$ and $\mu_{+}(z,Q)$ are 3 times continuously differentiable, and $\lvert\mu_{-}^{v}(z,Q)\lvert \leq U(Q)$ and $\lvert\mu^{v}_{+}(z,Q)\lvert \leq U(Q)$ for $v = 1,2,3$. 
\item $\sigma_{-}^{2}(z,Q)$ and $\sigma_{+}^{2}(z,Q)$ are continuous, and  $L(Q) \leq \sigma_{-}^{2}(z,Q) \leq U(Q)$ and $L(Q) \leq \sigma_{+}^{2}(z,Q) \leq U(Q)$. 
\end{enumerate}
\end{assumption}

In this setting, our parameter of interest is the average treatment effect (ATE) at the threshold,
\begin{align}\label{eq:PoI}
\theta(Q) &= \mu_{+}(0,Q) - \mu_{-}(0,Q)~.
\end{align}
The above parameter is identified, as shown by \cite{HTV01}, using the distribution of the observed random variables by
\begin{align}\label{eq:PoIobs}
  \theta(P) = \lim_{z \to 0^{+}}\mu(z, P)  - \lim_{z \to 0^{-}}\mu(z, P)~,
\end{align}
where $\mu(z,P) = E_{P}[Y|Z=z]$. The hypotheses of interest can then be stated as
\begin{align}\label{eq:HTn}
H_0: P \in \mathbf{P}_{0}~ \text{   versus   }~
H_1: P \in \mathbf{P}_{1} = \mathbf{P} \setminus \mathbf{P}_{0}~,
\end{align}
where $\mathbf{P}_{0} = \left\{ P \in \mathbf{P} ~ |~ \theta(P) = \theta_0  \right\}$ is the subset of $\mathbf{P}$ for which the null hypothesis that the ATE at the threshold equals a pre-specified value of $\theta_0$ holds. 

\begin{remark}
To be concise, we focus on the ATE in the so-called sharp RDD (characterized by the treatment assignment rule in \eqref{eq:RDDA}). Our results in Section \ref{sec:results} will however follow with some manipulation for other parameters such as quantiles, and for parameters in other designs such as the kink RDD in \cite{CLW12} or the fuzzy RDD. 
\end{remark}

\section{Testing Problem}\label{sec:uniform}

The testing problem we study is to ideally construct a finite sample test $\phi = \phi(X^{(n)})$ for \eqref{eq:HTn}. A requirement of the test is that it controls size, which is said to be level $\alpha$ whenever
\begin{align}\label{eq:size}
\sup_{P \in \mathbf{P}_{0}} E_{P^n}\left[ \phi \right] \leq \alpha ~,
\end{align}
where $\alpha \in (0,1)$ is the chosen level of significance. Note that the above is a finite sample requirement, and to construct nontrivial tests that control size in finite samples may be too demanding. Alternatively, we also study the construction of a sequence of tests $\{ \phi_n \}_{n=1}^{\infty}$ that are required to control limiting size, i.e.
\begin{align}\label{eq:UAS}
\limsup_{n \to \infty}\sup_{P \in \mathbf{P}_{0}} E_{P^n}\left[ \phi_{n} \right] \leq \alpha ~,
\end{align}
and are said to be uniformly asymptotically level $\alpha$. As highlighted in Remark \ref{rem:point}, this requirement is in contrast to the one for pointwise asymptotically valid tests where \eqref{eq:UAS} is not required to hold uniformly across distributions in $\mathbf{P}_0$. 

In our results, we show that under the commonly imposed setup described in the previous section, it is impossible to construct nontrivial tests that satisfy \eqref{eq:size} or sequence of such asymptotically nontrivial tests that satisfy \eqref{eq:UAS}. We achieve this by illustrating that \eqref{eq:HTn} has the property such that for any test $\phi$ the power against any alternative is bounded above by its size, i.e.
\begin{align}\label{eq:testability}
\sup_{P \in \mathbf{P}_1} E_{P^n}\left[ \phi \right] \leq \sup_{P \in \mathbf{P}_0} E_{P^n}\left[ \phi \right] ~.
\end{align}
To prove this claim, we rely on an insightful result from \cite{R04} restated in the following lemma for clarity, where 
\begin{align}
\tau(P,P') = \sup_{\{g:|g| \leq 1\}} \left\lvert \int gdP - \int gdP' \right\rvert ~
\end{align}
denotes the total variation metric between any two distributions $P$ and $P'$. This lemma additionally formalizes the concept of what we mean by $\mathbf{P}$ (and hence $\mathbf{Q}$) being large in some sense.

\begin{lemma}\label{lem:ML}
Let $n \geq 1$ and $\phi$ be any test of $\mathbf{P}_0$ versus $\mathbf{P}_1$ in \eqref{eq:HTn}. If for every $P \in \mathbf{P}_1$ there exists a sequence $\{P_{k}\}_{k=1}^{\infty}$ in $\mathbf{P}_{0}$ such that $\tau(P,P_{k}) \to 0$ as $k \to \infty$, then
\begin{align}\label{eq:Rlem}
\sup_{P \in \mathbf{P}_1} E_{P^n}\left[ \phi \right] \leq \sup_{P \in \mathbf{P}_0} E_{P^n}\left[ \phi \right]~.
\end{align}
\end{lemma}

\section{Main Results}\label{sec:results}

\subsection{Testability in the Basic Setup}\label{sec:NT}

In the following theorem we establish that when $\mathbf{Q}$ is defined as in \eqref{eq:Qdef}, any test for \eqref{eq:HTn} will have power against any alternative bounded above by its size.
\begin{theorem}\label{th:1}
Let $n \geq 1$, $\mathbf{Q}$ be defined as in \eqref{eq:Qdef}, $\mathbf{P}$ be as in \eqref{eq:push} and $\mathbf{P}_0$ and $\mathbf{P}_1$ be as in \eqref{eq:HTn}. Then any test $\phi$ satisfies
\begin{align}\label{eq:result}
\sup_{P \in \mathbf{P}_1} E_{P^n}\left[ \phi \right] \leq \sup_{P \in \mathbf{P}_0} E_{P^n}\left[ \phi \right] ~.
\end{align}
\end{theorem}

\begin{proof}
Fix $P \in \mathbf{P}_1$ and take any strictly positive sequence $\{\epsilon_k\}_{k=1}^{\infty}$ such that $\epsilon_k \to 0$ as $k \to \infty$. Since \eqref{eq:push} implies that $P= Q M^{-1}$ for some $Q \in \mathbf{Q}$, it then follows from Assumption \ref{ass:PC} (i) that for every $k$ there exists a Borel set $B_k$ in $\mathcal{X}$,
\begin{align}\label{eq:Aset}
  B_k = \{(y,z) \in \mathcal{X} : z \in (-\tilde{\epsilon}_k, \tilde{\epsilon}_k) \}~,
\end{align}
where $\tilde{\epsilon}_k > 0$, such that $0 < P(B_k) \leq \epsilon_k$. Take next any $P' \in \mathbf{P}_0$ that has the same density on $Z$ as $P$. We may then construct the sequence $\{P_k\}_{k=1}^{\infty}$ such that for every Borel subset $B$ of $\mathcal{X}$ let
\begin{align}\label{eq:Pk}
  P_k(B) := P(B \cap B_k^c) + P'(B \cap B_k)~, 
\end{align}
where $B_k^c$ denotes the complement of $B_k$. One can verify that for every $k$ that $P_k$ is a well defined distribution.

Next, we show that $\{P_k\}_{k=1}^{\infty}$ is in $\mathbf{P}_0$, i.e. for every $k$ there exists $Q_k \in \mathbf{Q}$ such that $\theta(Q_k) = \theta_0$ and $P_k = Q_kM^{-1}$. To construct this $Q_k$, first note that $P = Q M^{-1}$ and $P' = Q' M^{-1}$ for some $Q \in \mathbf{Q}$ and $Q' \in \mathbf{Q}$ with $\theta(Q') = \theta_0$. Then for every Borel subset $\tilde{B}$ of $\mathcal{W}$ let
\begin{align}\label{eq:Qk}
  Q_k(\tilde{B}) := Q(\tilde{B} \cap \tilde{B}_{k}^{c}) + Q'(\tilde{B} \cap \tilde{B}_{k})~, 
\end{align}
where 
\begin{equation}
  \tilde{B}_{k} = M^{-1}(B_k) = \{ (y_0,y_1,z) \in \mathcal{W} : z \in (-\tilde{\epsilon}_k, \tilde{\epsilon}_k)  \}~,
\end{equation} and $\tilde{B}_{k}^c$ denotes the complement of $\tilde{B}_{k}$, which in this case is just $M^{-1}(B^{c}_{k})$. Analogous to $P_k$, it follows that $Q_k$ is a well defined distribution. To show that $Q_k \in \mathbf{Q}$, first note \eqref{eq:Qk} ensures that $Q_k(A)=Q'(A)$ for every Borel subset $A$ of $\mathcal{W}$ that satisfies $A \subseteq \tilde{B}_{k}$. This implies that the density and all the conditional on $Z=z$ quantities in Assumption \ref{ass:PC} are equal for $Q_k$ and $Q'$ for all $z \in (-\tilde{\epsilon}_k, \tilde{\epsilon}_k)$. In turn, it follows that $Q_k$ satisfies Assumption \ref{ass:PC} by taking $\kappa(Q_k) = \min\{\kappa(Q'), \tilde{\epsilon}_k \}$, $L(Q_k)=L(Q')$, and $U(Q_k) = U(Q')$. Further, by the same argument, it follows that $\theta(Q_k) = \theta_0$ as $\theta(Q') = \theta_0$. Finally, given that for every $B \subseteq \mathcal{X}$ and $\tilde{B} = M^{-1}(B)$ we have $\tilde{B} \cap \tilde{B}_{k}^{c} = M^{-1}(B \cap B_k^c)$ and $\tilde{B} \cap \tilde{B}_{k} = M^{-1}(B \cap B_k)$, we can establish $P_k = Q_kM^{-1}$ from \eqref{eq:Pk} and \eqref{eq:Qk}. 

To conclude, we show that the total variation distance between $P$ and $P_k$ goes to 0 as $k \to \infty$,
\begin{align}
  \tau(P, P_k) &= \sup_{\{g:|g| \leq 1\}}\left\lvert \int g dP - \int g dP_k \right\lvert = \sup_{\{g:|g| \leq 1\}}\left\lvert \int_{B_k} g dP - \int_{B_k} g dP' \right\lvert \nonumber  \\
                      &\leq \left\lvert \int_{B_k} dP \right\lvert + \left\lvert\int_{B_k} dP' \right\lvert  \leq 2\epsilon_k \to 0~,
\end{align}
where the second and fourth relations follow from \eqref{eq:Pk} and \eqref{eq:Aset} respectively, along with noting that $P$ and $P'$ have the same density for $Z$. Since $P \in \mathbf{P}_1$ was chosen arbitrarily, we can then invoke Lemma \ref{lem:ML} to conclude the proof.
\end{proof}

\begin{remark}
In invoking Lemma \ref{lem:ML} to prove Theorem \ref{th:1}, for any $P \in \mathbf{P}_1$ we construct a sequence $\{P_k\}_{k=1}^{\infty}$ in $\mathbf{P}_0$, such that for every $k$ there exists a Borel set in $\mathcal{X}$ with positive probability under $P_k$ where $P$ and $P_k$ differ, and are otherwise equal on the complement of this set. Letting the probability of this set vanish with $k$ implies that $\tau(P,P_k) \to 0$ as $k \to \infty$. Further, since Assumption \ref{ass:PC} only requires conditions local to zero that are pointwise in nature, we formally show in the proof that this ensures that our construction $\{P_k\}_{k=1}^{\infty}$ falls in $\mathbf{P}_0$. Note that our construction is not unique and that multiple others are possible.
\end{remark}

\begin{remark}
It is important emphasize that Theorem \ref{th:1} is not a criticism of a specific test but holds for any choice of test. Furthermore, it is a statement on the finite sample property of any test, but with important asymptotic implications. To be specific, for any sequence of tests $\{ \phi_n \}_{n=1}^{\infty}$ with nontrivial limiting power, it directly follows from \eqref{eq:result} that
 \begin{equation}
    \limsup_{n \to \infty} \sup_{P \in \mathbf{P}_0} E_{P^n} \left[ \phi_{n} \right] > \alpha~.
\end{equation}  
This additionally implies that if the sequence of tests is pointwise consistent in power, i.e. pointwise power converges to one, then limiting size is in fact equal to one.
\end{remark}

\begin{remark}\label{rem:point}
Currently used tests are shown to be only pointwise asymptotically valid, i.e.
\begin{equation}\label{eq:point}
  \limsup_{n \to \infty} E_{P^n} \left[ \phi_{n} \right] \leq \alpha \text{ for all } P \in \mathbf{P}_0~,
\end{equation}
which does not say anything about whether this sequence of tests $\{ \phi_n \}_{n=1}^{\infty}$ approximates \eqref{eq:size} for large enough $n$. To be specific, it
is possible that for every $n \geq 1$ there exists $P \in \mathbf{P}_0$ such that 
\begin{equation}
E_{P^n}[\phi_n] >> \alpha~.
\end{equation}\end{remark}
\vspace{-0.45cm}
\subsection{Uniformly Valid Test under Stronger Assumptions}\label{sec:UTRDD}

In this section, we ask under what alternative assumptions we can construct an uniformly asymptotically valid test. We consider, in particular, a natural strengthening of Assumption \ref{ass:PC} leading to the following alternative definition of the set of possible distributions,
\begin{equation}\label{eq:Qdef2}
  \mathbf{Q} = \{ Q \in \mathbf{Q}_{\mathcal{W}} : Q \text{ satisfies Assumption \ref{ass:2}} \}~,
\end{equation}
where as before $\mathbf{Q}_{\mathcal{W}}$ denotes the set of all Borel probability measures on $\mathcal{W}$ that have a density on $Z$ with respect to the Lebesgue measure, and Assumption \ref{ass:2} is stated below. Note that if $Q$ satisfies Assumption \ref{ass:2} then it satisfies Assumption \ref{ass:PC}, and hence the definition of $\mathbf{Q}$ in \eqref{eq:Qdef2} generates a smaller set of distributions than the definition of $\mathbf{Q}$ in \eqref{eq:Qdef}.
\begin{assumption}\label{ass:2}
Let $Q$ be such that it satisfies Assumption \ref{ass:PC} with $\kappa(Q) = \tilde{\kappa}$, $L(Q) = \tilde{L}$ and $U(Q) = \tilde{U}$, where $\tilde{\kappa} > 0$, $\tilde{L} > 0$ and $\tilde{U} > 0$ are real numbers that do not depend on $Q$.
\end{assumption}

We next briefly describe a simple version of the \cite{CCT14} test (referred to as CCT hereafter), which is demonstrated to satisfy \eqref{eq:UAS} under this smaller set of distributions. For the null hypothesis in \eqref{eq:HTn}, the CCT test statistic is
\begin{align}\label{eq:TS}
T_n^{CCT}(X^{(n)}) = \frac{\hat{\theta}_n - \theta_0}{\hat{S}_n}~,
\end{align}
where $\hat{\theta}_n$ is a bias corrected local linear estimator of $\theta(P)$, and $\hat{S}_n$ is a plug-in estimator of a novel standard error formula that accounts for the variance of the bias estimate. The bias is estimated using a local quadratic estimator. Furthermore, for all estimates, we use the triangular kernel and a deterministic sequence of bandwidth choices denoted by $h_n$. Then, the CCT level $\alpha$ test is 
\begin{align}\label{eq:CCTP}
\phi_n^{CCT}(X^{(n)}) = 
\begin{cases}
& 1 \text{ if } |T_n^{CCT}(X^{(n)})| > z_{1 - \alpha/2} \\
& 0 \text{ otherwise}
\end{cases}~,
\end{align}
where $z_{1 - \alpha/2}$ is the $(1-\alpha/2)$-quantile of the standard normal distribution. 

The following theorem demonstrates that under the alternative definition of $\mathbf{Q}$ in \eqref{eq:Qdef2}, the test statistic in \eqref{eq:TS} for \eqref{eq:HTn} asymptotically converges uniformly in $\mathbf{P}_0$ to the standard normal distribution. It then directly follows that the test in \eqref{eq:CCTP} is uniformly asymptotically level $\alpha$, and, in fact, has limiting size equal to $\alpha$.

\begin{theorem}\label{th:UN}
Let $\mathbf{Q}$ be defined as in \eqref{eq:Qdef2}, $\mathbf{P}$ be as in \eqref{eq:push} and $\mathbf{P}_0$ and $\mathbf{P}_1$ be as in \eqref{eq:HTn}. If $nh_n \to \infty$, $h_n \to 0$ and $nh_n^{7} \to 0$, then the CCT test statistic from \eqref{eq:TS} satisfies
\begin{align}
T^{CCT}_n(X^{(n)}) = \frac{\hat{\theta}_n - \theta_0}{\hat{S}_n} \xrightarrow{d} \mathcal{N}(0,1)
\end{align}
as $n \to \infty$, where $X^{(n)}$ are i.i.d $P_n$ and $P_n$ is any sequence of distributions such that $P_n \in \mathbf{P}_0$ for all $n \geq 1$. This in turn implies that $\{ \phi_n^{CCT} \}_{n=1}^{\infty}$ in \eqref{eq:CCTP} is uniformly asymptotically level $\alpha$,  and, in fact, has limiting size equal to $\alpha$, i.e.
\begin{align}
\limsup_{n \to \infty} \sup_{P \in \mathbf{P}_0} E_{P^n} [ \phi_{n}^{CCT} ] = \alpha~.
\end{align}
\end{theorem}
The proof of the above essentially requires slightly altering the proof of the pointwise result in \cite{CCT14} to any sequence of distributions $P_n$ such that $P_n \in \mathbf{P}_0$ for all $n \geq 1$. For completeness, we provide a proof in the supplement appendix.

\begin{remark}
Note that when $\mathbf{Q}$ is defined as in \eqref{eq:Qdef2} the arguments used to prove Theorem \ref{th:1} do not go through. In particular, the constructed sequence $\{P_k\}_{k=1}^{\infty}$ in \eqref{eq:Pk} will not fall in $\mathbf{P}$, as the corresponding $\{Q_k\}_{k=1}^{\infty}$ in \eqref{eq:Qk} will not fall in $\mathbf{Q}$. To see why, note that for large enough $k$ we have $\kappa(Q_k) < \tilde{\kappa}$, and either $\mu_{-}(z,Q_k)$ or $\mu_{+}(z,Q_k)$ is discontinuous at either $z= -\kappa(Q_k)$ or $z= \kappa(Q_k)$. This implies $Q_k$ will not satisfy Assumption \ref{ass:2} as $\mu_{-}(z,Q_k)$ or $\mu_{+}(z,Q_k)$ will not be continuous for all $z \in (-\tilde{\kappa},\tilde{\kappa})$,. Intuitively, Assumption \ref{ass:2} excludes extreme sequences such as $\{Q_k\}_{k=1}^{\infty}$ for which nonparametric tools work poorly to give a uniform limit result. For recent additional results on uniform testing in RDD, see \cite{AK16} and \cite{CCF16}.
\end{remark}

\clearpage
\bibliography{references}
\nocite{*}

\clearpage

\end{document}

%\date{}

\author{
Vishal Kamat\\
Departments of Economics\\
Northwestern University\\
\url{v.kamat@u.northwestern.edu}
}

\title{Supplement Appendix to ``On Nonparametric Inference in the Regression Discontinuity Design''}

\maketitle

\begin{abstract}
This documents provides a proof to Theorem 4.2 in the author's paper ``On Nonparametric Inference in the Regression Discontinuity Design''.
\end{abstract}

\noindent KEYWORDS: Regression discontinuity design, uniform testing.

\noindent JEL classification codes: C12, C14.

\appendix
\renewcommand{\theequation}{\Alph{section}-\arabic{equation}}
\small
\section{Additional Notation}

Let $Z^{(n)} = \{ Z_i : 1 \leq i \leq n\}$ denote the observed sample of the random variable $Z$. Let $a_n \precsim b_n$ denote $a_n \leq A b_n$, where $a_n$ and $b_n$ are deterministic sequences and $A$ is a positive constant uniform in $\mathbf{P}$. Let $| \cdot|$ denote the Euclidean matrix norm. As we use the notion of convergence in probability under the sequence of distributions $P_n$, let $A_n = o_{P_n}(1)$ denote
\begin{equation*}
   P_n(|A_n| > \epsilon) \to 0 \text{ as } n \to \infty~,
\end{equation*}
for a sequences of random variables $A_n \sim P_n$, where $\epsilon$ is any constant such that $\epsilon > 0$. Further, in Table \ref{tab:notation} below, we introduce additional notation to keep our arguments concise.
%\begin{table}[H]
%\centering
%\begin{tabular}{r l}
\begin{longtable}{r l}
$H(h_n)$ & diag(1,$h_n^{-1}$,$h_n^{-2}$) \\
$r(Z_i/h_n)$ & $(1,Z_i/h_n,(Z_i/h_n)^2)'$ \\
$Z_n(h_n)$   & $(r(Z_1/h_n), \ldots, r(Z_n/h_n))'$ \\
$k(u)$          & $(1-u)1\{ 0 \leq u \leq 1  \}$ \\
$K(u)$          & $k(-u) 1\{ u < 0\} + k(u) 1\{ u \geq 0 \}$ \\
$K_{h_n}(u)$ & $K(u/h_n)/h_n$ \\
$W_{+,n}(h_n)$  & diag$\left(1\{Z_1 \geq 0\}K_{h_n}(Z_1), \ldots, 1\{Z_n \geq 0\}K_{h_n}(Z_n)\right)$  \\
$\Gamma_{+,n}(h_n)$ & $Z_n(h_n)'W_{+,n}(h_n)Z_n(h_n)/n$ \\
$S_n(h_n)$       & $((Z_1/h_n)^3, \ldots, (Z_n/h_n)^3)'$ \\
$\nu_{+,n}$      & $Z_n(h_n)'W_{+,n}(h_n)S_n(h_n)/n$  \\
$e$              & $(1,0,0)'$  \\
$\mu(z,P)$  & $E_{P}[Y|Z=z]$ \\
$\mu_+(P)$ &  $\lim_{z \to 0^+}\mu(z,P)$ \\
$\mu_-(P)$ & $\lim_{z \to 0^-}\mu(z,P)$ \\
$\mu^v(z,P)$ & $d^v \mu(z,P)/dz^v$ \\
$\mu^v_+(P)$ & $\lim_{z \to 0^+} \mu^v(z,P)$ \\
$\sigma^2(z,P)$ & $Var_{P}[Y | Z=z]$ \\
$\Sigma_n(P)$  & diag$(\sigma^2(Z_1,P), \ldots, \sigma^2(Z_n,P))$\\
$\Psi_{+,n}(h_n,P)$ & $Z_n(h_n)'W_{+,n}(h_n) \Sigma_n(P) W_{+,n}(h_n) Z_n(h_n) / n$ \\
$\mathbf{Y}_n$              & $(Y_1, \ldots, Y_n)'$ \\
$\hat{\beta}_{+,n}$ & $H(h_n) \Gamma^{-1}_{+,n}(h_n) Z_n(h_n)'W_{+,n}(h_n) \mathbf{Y}_n /n$ \\
%\end{tabular}
\caption{Important Notation}
\label{tab:notation}
\end{longtable}
%\end{table}

Next, we provide an extended description of the test statistic used. For our null hypothesis as stated in the paper, the test statistic can be rewritten as
\begin{equation}\label{eq:TS}
  T_n^{CCT} = \frac{\hat{\mu}_{+,n} + \hat{\mu}_{-,n} - \left( \mu_+(P) - \mu_-(P) \right)}{\hat{S}_n}~,
\end{equation}
where $\mu_+(P) - \mu_-(P) = \theta_0$, $\hat{\mu}_{+,n}$ and $\hat{\mu}_{-,n}$ are bias corrected local linear estimates of $\mu_+(P)$ and $\mu_-(P)$, and $$\hat{S}_n = \sqrt{\hat{V}_{+,n} + \hat{V}_{-,n}}~,$$ where $\hat{V}_{+,n}$ and $\hat{V}_{+,n}$ are plug-in estimated conditional on $Z^{(n)}$ variances of $\hat{\mu}_{+,n} $ and $\hat{\mu}_{-,n}$; see \eqref{eq:vest} for the plug-in estimator used. The bias  of both estimates are estimated using local quadratic estimators. Furthermore, for all estimates, we use the triangular kernel, i.e. $k(u)$ in Table \ref{tab:notation}, and a deterministic sequence of bandwidth choices denoted by $h_n$. Throughout this document, we provide results for quantities with subscript $(+)$ as arguments for those with subscript $(-)$ follow symmetrically. In addition, as noted in \citet[][Remark 7]{CCT14}, we exploit the fact that in our simple version of the test statistic the estimates are numerically equivalent to those from a non bias corrected local quadratic estimator. In turn, we can write
\begin{equation}
  \hat{\mu}_{+,n} = e' \hat{\beta}_{+,n}~,
\end{equation}
which reduces the length of the proof presented below. Further, as stated in the paper, note that
\begin{align}\label{eq:Qdef}
  \mathbf{Q} = \{ Q \in \mathbf{Q}_{\mathcal{W}} : Q \text{ satisfies Assumption 4.1} \}~,
\end{align}
and that
\begin{align}\label{eq:Pdef}
  \mathbf{P} = \{ QM^{-1} : Q \in \mathbf{Q}\}~,
\end{align}
where $\mathbf{Q}_{\mathcal{W}}$, $M^{-1}$ and Assumption 4.1 are as defined in the paper.

\section{Auxiliary Lemmas}

\begin{lemma}\label{lem:quantities}
 Let $\mathbf{Q}$ be defined as in \eqref{eq:Qdef}, $\mathbf{P}$ be as in \eqref{eq:Pdef} and $P_n \in \mathbf{P}$ for all $n \geq 1$. If $nh_n \to \infty$ and $h_n \to 0$, then
 \begin{enumerate}[(i)]
   \item $\Gamma_{+,n}(h_n) = \tilde{\Gamma}_{+,n}(h_n) + o_{P_n}(1)$, where $\tilde{\Gamma}_{+,n}(h_n) = \int_{0}^{\infty} K(u) r(u) r(u)' f_{P_n}(uh_n) du \in [\Gamma_L, \Gamma_U]$ .
   \item $\nu_{+,n}(h_n) = \tilde{\nu}_{+,n}(h_n) + o_{P_n}(1)$, where $\tilde{\nu}_{+,n}(h_n) = \int^{\infty}_{0} K(u) r(u) u^2 f_{P_n}(uh_n) du \in [\nu_L , \nu_U]$.
   \item $h_n \Psi_{+,n}(h_n,P_n) = \tilde{\Psi}_{+,n}(h_n) + o_{P_n}(1)$, where $\tilde{\Psi}_{+,n}(h_n) = \int^{\infty}_{0} K(u)^2 r(u) r(u)' \sigma^2_{P_n}(uh_n) f_{P_n}(uh_n) du \in [\Psi_L, \Psi_U]$.
 \end{enumerate}
\end{lemma}
\begin{proof}
For (i), a change of variables gives us
\begin{align*}
  E_{P_n^n}[\Gamma_{+,n}(h_n)] &= E_{P_n}\left[ \frac{1}{nh_n} \sum^{n}_{i=1} 1\{ Z_i \geq 0\} K(Z_i / h_n) r(Z_i / h_n) r(Z_i / h_n)' \right] \\
      & = \frac{1}{h_n} \int^{\infty}_{0} K(z/h_n) r(z/h_n)r(z/h_n)'f_{P_n}(z) dz \\
      & = \int_{0}^{\infty} K(u) r(u) r(u)' f_{P_n}(uh_n) \equiv \tilde{\Gamma}_{+,n}(h_n)~.
\end{align*}
Further, since $h_n < \tilde{\kappa}$ for large enough $n$, we have that $\tilde{L} \leq f_{P_n}(z) \leq \tilde{U}$ by Assumption 4.1, which implies that
\begin{align*}
  \Gamma_L \equiv \tilde{L} \int_{0}^{\infty} K(u) r(u) r(u)' du \leq \tilde{\Gamma}_{+,n}(h_n) \leq \tilde{U} \int_{0}^{\infty} K(u) r(u) r(u)' du \equiv \Gamma_U~,
\end{align*}
and that
\begin{align*}
  E_{P_n^n}[ | \Gamma_{+,n}(h_n) - E_{P_n}[\Gamma_{+,n}(h_n)] |^2] & \leq \frac{1}{h_n^2} E_{P_n}\left[\left|  1\{Z_i \geq 0\} K(Z_i / h_n) r(Z_i / h_n) r(Z_i / h_n)' \right|^2 \right] \\
    & = \frac{1}{n h_n}\int^{\infty}_{0} K(u)^2 |r(u)|^4 f_{P_n}(uh_n) du \\
    & \leq \frac{\tilde{U}}{n h_n}\int^{\infty}_{0} K(u)^2 |r(u)|^4  du \\
    & = O(n^{-1}h_n^{-1}) = o(1)~.
\end{align*}
It then follows by Markov's Inequality that $\Gamma_{+,n}(h_n) = \tilde{\Gamma}_{+,n}(h_n) + o_{P_n}(1)$. Analogously, closely following \citet[][Lemma S.A.1]{CCT14supp}, we can show Lemma \ref{lem:quantities}(ii)-(iii).
\end{proof}

\begin{lemma}\label{lem:terms}
   Let $\mathbf{Q}$ be defined as in \eqref{eq:Qdef}, $\mathbf{P}$ be as in \eqref{eq:Pdef} and $P_n \in \mathbf{P}$ for all $n \geq 1$. If $nh_n \to \infty$ and $h_n \to 0$, then
  \begin{enumerate}[(i)]
   \item  $E_{P_n^n}[\hat{\mu}_{+,n} | Z^{(n)}] = \mu_{+}(P_n) + h_n^3 e' \Gamma^{-1}_{+,n}(h_n) \nu_{+,n}(h_n) \mu_{+}^3(P_n) / 6 + o_{P_n}(h_n^3)~$.
   \item $V_{P_n^n}[\hat{\mu}_{+,n} | Z^{(n)}] = n^{-1} e' \Gamma^{-1}_{+,n}(h_n) \Psi_{+,n}(h_n,P_n) \Gamma_{+,n}^{-1}(h_n)e \equiv V_{+,n}(h_n,P_n)~$.
   \item $\left( V_{+,n}(h_n,P_n)\right)^{-1/2} \left(\hat{\mu}_{+,n} - E_{P_n^n}[\hat{\mu}_{+,n} | Z^{(n)}]\right) \xrightarrow{d} \mathcal{N}(0,1)~$.
  \end{enumerate}
\end{lemma}
\begin{proof}
For (i), by taking the conditional on $Z^{(n)}$ expectation, we have
\begin{align*}
  E_{P_n^n}[\hat{\mu}_{+,n} | Z^{(n)}] &= e' H(h_n) \Gamma^{-1}_{+,n}(h_n) Z_n(h_n)' W_{+,n}(h_n) E_{P^n_n}[\mathbf{Y}_n | Z^{(n)} ] / n~.
\end{align*}  
Further, as $h_n < \tilde{\kappa}$ for large enough $n$, we have by the required differentiability in Assumption 4.1 and a taylor expansion that
\begin{align*}
  E_{P_n^n}[\mathbf{Y}_n | Z^{(n)}] / n = Z_n(h_n) H(h_n)^{-1} \beta_+(P_n) / n + h_n^3 S_n(h_n) \mu^3_+(P_n) /(6n) + o_{P_n}(h_n^3)~,
\end{align*}
where $\beta_+(P_n) = (\mu_+(P_n),\mu^1_+(P_n),\mu^2_+(P_n)/2)'$. It then follows from Lemma \ref{lem:quantities} that
\begin{align*}
  E_{P_n^n}[\hat{\mu}_+ | Z^{(n)}] = \mu_{+}(P_n) + h_n^3 e \Gamma^{-1}_{+,n}(h_n) \nu_{+,n}(h_n) \mu_{+}^3(P_n) / 3! + o_{P_n}(h_n^3)~.
\end{align*}

For (ii), a simple calculation gives us
\begin{align*}
  V_{P_n^n}[\hat{\mu}_{+}(h_n) | Z^{(n)}] &= e' H(h_n) \Gamma^{-1}_{+,n}(h_n) Z_n(H_n)' W_{+,n}(h_n) \Sigma_n(P_n) W_{+,n}(h_n) Z_n(h_n) \Gamma^{-1}_{+,n}(h_n) H(h_n) e / n^2 \\
    & = n^{-1} e' \Gamma^{-1}_{+,n}(h_n) \Psi_{+,n}(h_n,P_n) \Gamma^{-1}_{+,n}(h_n) e \equiv V_{+,n}(h_n,P_n)~.
\end{align*}

For (iii), first note that from Lemma \ref{lem:quantities} we have $V_{+,n}(h_n,P_n) = \tilde{V}_{+,n}(h_n) + o_{P_n}(1)~,$ where
\begin{align*}
\tilde{V}_{+,n}(h_n) = (nh_n)^{-1} e' \tilde{\Gamma}^{-1}_{+,n}(h_n) \tilde{\Psi}_{+,n}(h_n) \tilde{\Gamma}^{-1}_{+,n}(h_n) e~.
\end{align*}
Then rewrite as follows
\begin{align}\label{eq:LFrewrite}
  \frac{\hat{\mu}_{+,n} - E_{P_n^n}[\hat{\mu}_{+,n} | Z^{(n)}]}{\sqrt{V_{+,n}(h_n,P_n)}} = \left( \frac{\tilde{V}_{+,n}(h_n,P_n)}{V_{+,n}(h_n,P_n)} \right)^{1/2} \left( \tilde{V}_{+,n}(h_n) \right)^{-1/2} e' \Gamma^{-1}_{+,n}(h_n) \tilde{\Gamma}_{+,n}(h_n)\tilde{A}_n^{1/2}  \xi_{n}~,
\end{align}
where 
\begin{align*}
\xi_{n}          &=\sum^{n}_{i=1} \omega_{n,i} \epsilon_{n,i} ~, \\
\epsilon_{n,i}     &= Y_i - E_{P_n}[Y_i|Z_i]~, \\
\tilde{A}_n   &= (nh_n)^{-1} \tilde{\Gamma}^{-1}_{+,n}(h_n) \tilde{\Psi}_{+,n}(h_n) \tilde{\Gamma}^{-1}_{+,n}(h_n) ~, \text{ and} \\
\omega_{n,i} &=  \tilde{A}_n^{-1/2} \tilde{\Gamma}^{-1}_{+,n}(h_n) K_{h_n}(Z_i/h_n) r(Z_i/h_n) / n~.
\end{align*}
Next note that for any $a \in \mathbf{R}^3$ we have that $\{a' \omega_{n,i} \epsilon_{n,i} : 1 \leq i \leq n\}$ is a triangular array of independent random variables with $E_{P_n^n}[a' \xi_{n}] = 0$ and $V_{P_n^n}[a' \xi_n] = a'a$. Further, this triangular array satisfies the Lindeberg condition. To see why, first note that by Lemma \ref{lem:quantities} we have
\begin{align}\label{eq:varbound}
 |\tilde{A}_{n}| \geq (nh_n)^{-1}|\tilde{A}_L|~,
\end{align}
for some value $\tilde{A}_{L} \in \mathbf{R}$, which is uniform in $\mathbf{P}$. We then have in addition to by Lemma \ref{lem:quantities} and a change of variables that
\begin{align*}
  \sum^{n}_{1=1}E_{P_n^n}[|\omega_{n,i} \epsilon_i|^4] &\precsim (nh_n)^2 \sum^{n}_{1=1}E_{P_n^n} \left[ \left|   K_{h_n}(Z/h_n) r(Z/h_n) / n  \right|^4  \right] \\
  & \precsim (nh_n)^2 n^{-3} h_n^{-4} \int^{\infty}_{0} \left|  K(z/h_n) r(z/h_n)  \right|^4 f_{P_n}(z) dz \\
  & \precsim  (nh_n)^2 n^{-3} h_n^{-3} = O\left( (nh_n)^{-1} \right) = o(1)~
\end{align*}
and hence, using the Lindeberg-Feller CLT, we have that $a' \xi_n \xrightarrow{d}  \mathcal{N}(0,a'a)$ as $n \to \infty$. Since this holds for any $a \in \mathbf{R}^3$, the Cramer-Wold theorem implies that $\xi_n \xrightarrow{d}  \mathcal{N}(0,I_3)$ as $n \to \infty$, where $I_3$ denotes the identity matrix of size three. Furthermore, analogous to $V_{+}(h_n,P_n) = \tilde{V}_{+}(h_n) + o_{P_n}(1)$, we can show that
\begin{align}\label{eq:Vratio}
  \frac{V_{+,n}(h_n,P_n)}{\tilde{V}_{+,n}(h_n)}  = 1 + o_{P_n}(1)~.
\end{align}
Further, by Lemma \ref{lem:quantities} we have that
\begin{align}
  \Gamma^{-1}_{+,n}(h_n) \tilde{\Gamma}_{+,n}(h_n)  = I_3 + o_{P_n}(1)~.
\end{align}
Substituting the above results in \eqref{eq:LFrewrite} concludes the proof.
\end{proof}

\section{Proof of Theorem 4.2}

Here we show only that
\begin{align*}
  \frac{\hat{\mu}_{+,n} - \mu_{+}(P_n)}{\sqrt{\hat{V}_{+,n}}} \xrightarrow{d} \mathcal{N}(0,1)~,
\end{align*}
since under similar arguments it will follow that
\begin{align*}
  \frac{\hat{\mu}_{n,-} - \mu_{-}(P_n)}{\sqrt{\hat{V}_{n,-}}} \xrightarrow{d} \mathcal{N}(0,1)~,
\end{align*}
and then due to independence we can conclude that $T_{n}^{CCT} \xrightarrow{d} \mathcal{N}(0,1)$ as $n \to \infty$. To this end, first rewrite
\begin{align*}
  \frac{\hat{\mu}_{+,n} - \mu_{+}(P_n)}{\sqrt{\hat{V}_{+,n}}} = \frac{\hat{\mu}_{+,n} - \mu_{+}(P_n)}{\sqrt{V_{+,n}(h_n,P_n)}} \cdot \sqrt{\frac{V_{+,n}(h_n,P_n)}{\hat{V}_{+,n}}}~.
\end{align*}

\underline{Step 1.} We show that 
\begin{align}\label{eq:step1}
  \frac{\hat{\mu}_{+,n} - \mu_{+}(P_n)}{\sqrt{V_{+,n}(h_n,P_n)}} \xrightarrow{d} \mathcal{N}(0,1)~.
\end{align}
To begin, first rewrite the above as
\begin{align*}
  \frac{\hat{\mu}_{+,n} - E_{P_n^n}[\hat{\mu}_{+,n} | Z^{(n)}]}{\sqrt{V_{+,n}(h_n,P_n)}} + \left( \frac{\tilde{V}_{+,n}(h_n)}{V_{+,n}(h_n,P_n)} \right)^{1/2} \frac{E_{P_n^n}[\hat{\mu}_{+,n} | Z^{(n)}] - \mu_{+}(P_n)}{\sqrt{\tilde{V}_{+,n}(h_n)}}~.
\end{align*}
In Lemma \ref{lem:terms} (iii), we showed that
\begin{align*}
  \frac{\hat{\mu}_{+,n} - E_{P_n^n}[\hat{\mu}_{+,n} | Z^{(n)}]}{\sqrt{V_{+,n}(h_n,P_n)}} \xrightarrow{d} \mathcal{N}(0,1)~\text{ and }~\frac{\tilde{V}_{+,n}(h_n)}{V_{+,n}(h_n,P_n)} = 1 + o_{P_n}(1)~.
\end{align*}
It then remains to show that
\begin{align*}
  \frac{E_{P_n^n}[\hat{\mu}_{+,n} | Z^{(n)}] - \mu_+(P_n)}{\sqrt{\tilde{V}_{+,n}(h_n)}} = o_{P_n}(1)~,
\end{align*}
to conclude. To this end, note that by Lemma \ref{lem:terms}, Lemma \ref{lem:quantities} and \eqref{eq:varbound}, it follows that
\begin{align*}
  \frac{\left| E_{P_n^n}[\hat{\mu}_{+,n} | Z^{(n)}] - \mu_+(P_n) \right|}{\sqrt{\tilde{V}_{+,n}(h_n)}} = O\left((nh_n)^{-1/2}\right) \left\{ O(h_n^3) + o_{P_n}(1) \right\} = O\left((nh_n^7)^{1/2}\right) + o_{P_n}(1) = o_{P_n}(1)
\end{align*}
as $h_n \to 0$, $nh_n \to \infty$ and $nh_n^7 \to 0$.

\underline{Step 2.} We show that
\begin{align}
  \frac{V_{+,n}(h_n,P_n)}{\hat{V}_{+,n}} = 1 + o_{P_n}(1)~.
\end{align}
To begin note that
\begin{align}
 nh \left( V_{+,n}(h_n,P_n) - \hat{V}_{+,n} \right)=  e' \Gamma^{-1}_{+,n}(h_n) \cdot  h\left( \Psi_{+,n}(h_n,P_n) - \hat{\Psi}_{+,n}(h_n) \right) \cdot \Gamma^{-1}_{+,n}(h_n) e~,
\end{align}
where
\begin{align}\label{eq:Vbig}
 h  \left(\Psi_{+,n}(h_n,P_n) - \hat{\Psi}_{+,n}(h_n)\right) = h Z_n(h_n)'W_{+,n}(h_n) \left(\Sigma_n(P_n) - \hat{\Sigma}_n \right) W_{+,n}(h_n) Z_n(h_n) / n ~,
\end{align}
and
\begin{align}\label{eq:vest}
  \hat{\Sigma}_{+,n}  = \text{diag}(\hat{\epsilon}^2_{+,n,1}, \ldots, \hat{\epsilon}^2_{+,n,n})~,
\end{align}
such that $\hat{\epsilon}_{+,n,i} = Y_i - \hat{\mu}_{+,n}$. Further, note that by construction, we can write
\begin{align}
  Y_i = \mu(Z_i,P_n) + \epsilon_{n,i}~,
\end{align}
such that $E_{P_n}[\epsilon_{n,i}] = 0$ and $Var_{P_n}[\epsilon_{n,i}|Z=z] = \sigma^{2}(z,P_n)$. This in turn implies
\begin{align}
  \hat{\epsilon}_{+,n,i} = \epsilon_{n,i} + \mu(Z_i,P_n) - \mu_+(P_n) + \mu_+(P_n) - \hat{\mu}_{+,n}~.
\end{align}
We can then expand \eqref{eq:Vbig} to get the following 
\begin{align}
  h \left(\Psi_{+,n}(h_n,P_n) - \hat{\Psi}_{+,n}(h_n) \right)=& \underbrace{h \sum^n_{i=1} 1\{ Z_i \geq 0 \} (\sigma^2(Z_i,P_n) - \epsilon_{n,i}^2) K_{h_n}(Z_i)^2 r(Z_i / h_n) r(Z_i / h_n)' / n}_{\equiv B_{1,n}~,~ \text{(a)}} \nonumber \\
  & - \underbrace{h \sum^n_{i=1} 1\{ Z_i \geq 0 \} ( \mu(Z_i,P_n) - \hat{\mu}_{+,n})^2 K_{h_n}(Z_i)^2 r(Z_i / h_n) r(Z_i / h_n)' / n}_{\equiv B_{2,n}~,~ \text{(b)}} \nonumber \\
  & + 2 \underbrace{h \sum^n_{i=1} 1\{ Z_i \geq 0 \} \epsilon_{n,i} ( \mu(Z_i,P_n) - \hat{\mu}_{+,n}) K_{h_n}(Z_i)^2 r(Z_i / h_n) r(Z_i / h_n)'/ n}_{\equiv B_{3,n}~,~ \text{(c)}}~. \nonumber
\end{align}
For quantity (a), since Assumption 2.1 (i), Assumption 2.1 (ii) and Assumption 2.1 (iv) are satisfied with the required uniform constants, we have by a change of variables that
\begin{align*}
  E_{P_n}\left[ | B_{1,n} |^2  \right] &\precsim (nh)^{-1}\int^{\infty}_{0} K(u)^4 |r(u)|^4  du \\
  &= O((nh)^{-1})  =  o(1)~,
\end{align*}
which implies by Markov's Inequality that $B_{n,1}=o_{P_n}(1)$. For quantity (b), note that first we can rewrite it as
\begin{align*}
  B_{n,2} =& \underbrace{h \sum^n_{i=1} 1\{ Z_i \geq 0 \} ( \mu(Z_i,P_n) - \mu_{+}(P_n))^2 K_{h_n}(Z_i)^2 r(Z_i / h_n) r(Z_i / h_n)' / n}_{\equiv B_{n,21}} \\
  & + ( \mu_{+}(P_n) - \hat{\mu}_{+,n})^2 \cdot \underbrace{h \sum^n_{i=1} 1\{ Z_i \geq 0 \} K_{h_n}(Z_i)^2 r(Z_i / h_n) r(Z_i / h_n)' / n}_{\equiv B_{n,22}} \\
  &+ 2 ( \mu_{+}(P_n) - \hat{\mu}_{+,n}) \cdot \underbrace{h \sum^n_{i=1} 1\{ Z_i \geq 0 \} ( \mu(Z_i,P_n) - \mu_{+}(P_n))  K_{h_n}(Z_i)^2 r(Z_i / h_n) r(Z_i / h_n)' / n}_{\equiv B_{n,23}}~,
\end{align*}
Next, since Assumption 2.1 (i) and Assumption 2.1 (iii) are satisfied with the required uniform constants, we have by a taylor approximation and a change of variables that
\begin{align*}
 E_{P_n}[ | B_{n,21}|^2] &\precsim n^{-1} h^3 \int^{\infty}_{0} K(u)^4 |r(u)|^4  du \\
  &= O(n^{-1} h^3)  =  o(1)~,
\end{align*}
which implies by Markov's inequality that $B_{n,21} = o_{P_n}(1)$. Further, since  Assumption 2.1 (i) is satisfied with the required uniform constants, we have by a change of variables that
\begin{align*}
  E_{P_n}[ | B_{n,22}|^2] &\precsim (nh)^{-1} \int^{\infty}_{0} K(u)^4 |r(u)|^4  du \\
  &= O((nh)^{-1})  =  o(1)~,
\end{align*}
which implies by Markov's inequality that $B_{n,22} = o_{P_n}(1)$. Finally, since Assumption 2.1 (i) and Assumption 2.1 (iii) are satisfied with the required uniform constants, we have by a taylor approximation and a change of variables that
\begin{align*}
  E_{P_n} \left[ |  B_{n,23} |^2\right] &\precsim (n)^{-1} h \int^{\infty}_{0} K(u)^4 |r(u)|^4  du \\
  &= O(n^{-1} h)  =  o(1)~,
\end{align*}
which implies by Markov's inequality that $B_{n,23} = o_{P_n}(1)$. Since $\mu_{+}(P_n) - \hat{\mu}_{+,n} = o_{P_n}(1)$ by \eqref{eq:step1}, we can conclude for quantity (b) that $B_{n,2} = o_{P_n}(1)$. For quantity (c), using analogous arguments, we can conclude that $B_{n,3} = o_{P_n}(1)$, and hence
\begin{align}
  h \left(\Psi_{+,n}(h_n,P_n) - \hat{\Psi}_{+,n}(h_n) \right) = o_{P_n}(1)~.
\end{align}
In addition, since from Lemma \ref{lem:quantities} we have that $\Gamma^{-1}_{+,n}(h_n) = \tilde{\Gamma}^{-1}_{+,n}(h_n)$, it then follows that
\begin{align}\label{eq:Vend}
  nh \left( V_{+,n}(h_n,P_n) - \hat{V}_{+,n} \right) = o_{P_n}(1)~.
\end{align}

To conclude, first rewrite \eqref{eq:Vend} as
\begin{align*}
  \frac{ V_{+,n}(h_n,P_n) - \hat{V}_{+,n}}{\tilde{V}_{+,n}(h_n) } = o_{P_n}(1)~,
\end{align*}
and our result then follows from \eqref{eq:Vratio}.

\clearpage
%\small
\bibliography{referencessupp}
\nocite{*}

\clearpage